\documentclass[prl,superscriptaddress,twocolumn,nobalancelastpage]{revtex4-2}
\usepackage{cmap}
\usepackage{amsmath}
\usepackage{amssymb}
\usepackage{amsfonts}
\usepackage{amsthm}
\usepackage{mathtools}
\usepackage{braket}
\usepackage[dvipsnames]{color,xcolor,colortbl}
\usepackage{blkarray}
\usepackage{nicefrac}
\usepackage{theoremref}
\usepackage{soul}
\usepackage{graphics,subcaption}
\usepackage[justification=raggedright,singlelinecheck=true]{caption}

\usepackage[normalem]{ulem}
\usepackage{xr-hyper}
\usepackage{hyperref}
\hypersetup{%
    colorlinks=true,
    linkcolor=blue,
    citecolor=blue,
    urlcolor=blue,
    plainpages=false
    }

\newtheorem{theorem}{Theorem}

\newtheorem{definition}{Definition}
\newcommand{\bx}[0]{\boldsymbol{x}}
\newcommand{\by}[0]{\boldsymbol{y}}
\newcommand{\bz}[0]{\boldsymbol{z}}
\newcommand{\bp}[0]{\boldsymbol{p}}
\newcommand{\br}[0]{\boldsymbol{r}}
\newcommand{\bs}[0]{\boldsymbol{s}}
\newcommand{\bn}[0]{\boldsymbol{n}}
\newcommand{\bl}[0]{\boldsymbol{l}}
\renewcommand{\bm}[0]{\boldsymbol{m}}
\newcommand{\bgamma}[0]{\boldsymbol{\gamma}}
\newcommand{\sfa}[0]{\mathsf{A}}
\newcommand{\sfb}[0]{\mathsf{B}}
\newcommand{\sfs}[0]{\mathsf{S}}

\newcommand\tr{\textnormal{tr}}
\def\prg#1{\paragraph*{{\bf #1}}}

\begin{document}

\title{Absolutely Maximal Entanglement in Continuous Variables}

\date{\today}
\author{James~I.~Kwon}
\affiliation{Department of Mathematics, University of Maryland, College Park, Maryland 20742, USA}
\author{Anthony~J.~Brady}
\affiliation{Joint Center for Quantum Information and Computer Science, NIST \& University of Maryland, College Park, Maryland 20742, USA}
\author{Victor~V.~Albert}
\affiliation{Joint Center for Quantum Information and Computer Science, NIST \& University of Maryland, College Park, Maryland 20742, USA}

\begin{abstract}
    We explore absolutely maximal entanglement (AME) and $k$-uniformity in continuous-variable (CV) quantum systems, and show that---unlike in qudit systems---such entanglement is readily realizable in both Gaussian and non-Gaussian quantum states of multiple modes. We demonstrate that Gaussian CV cluster states are generically AME, rederiving the results of [Phys. Rev. Lett. 103, 070501 (2009)] from a generalized stabilizer formalism, and provide explicit constructions based on Cauchy, Vandermonde, totally positive, and real-block-code generator matrices. 
    We further extend AME properties to a family of non-Gaussian states constructed from discrete Zak basis states that incorporate grid states (\textit{a.k.a.}, Gottesman-Kitaev-Preskill states) as non-Gaussian resources. Realizations of CV AME states enable open-destination multi-party CV teleportation, CV quantum secret sharing, CV majority-agreed key distribution, perfect-tensor networks on arbitrary geometries, and multi-unitary circuits. Our extension to non-Gaussian AME states may further provide robustness to Gaussian noise and benefits to quantum CV information processing.
\end{abstract}

\maketitle

Quantum entanglement is necessary for nearly all modern quantum applications.
A highly entangled many-party state cannot give away much information about itself when probed on a small subset of parties.
Absolutely maximal entanglement (AME)~\cite{PhysRevA.62.030301,scott,PhysRevA.77.060304,hsieh2011undetermined,PhysRevA.87.012319} quantifies the extreme case when the state is maximally mixed on all combinations of less than half of the parties, admitting maximal entanglement between each combination and its complement.

AME states cannot boost measurement-based quantum computation~\cite{PhysRevLett.102.190501,PhysRevLett.102.190502}, but they do yield a range of other applications, including
open-destination multi-party teleportation~\cite{HelwigTeleportation,HelwigSecret}, 
quantum secret sharing~\cite{HelwigTeleportation,BaiQSS,HelwigSecret}, 
majority-agreed key distribution~\cite{MAKD},
information masking~\cite{BaiQSS,ShiMasking,ModiMasking},
holographic tensor networks~\cite{Pastawski_2015,Hayden2016randHolographic}, 
multi-unitary circuits~\cite{BertiniProsenExact,Rather_2022,JonayTriunitary},
and quantum error-correcting codes~\cite{scott,PhysRevA.76.042309,raissi2018optimal,raissi2020modifying,huber2020quantum,PhysRevA.103.022402}.

Exact AME is uncommon.
AME states do not exist for several combinations of numbers and dimensions of parties~\cite{borras2007multiqubit,PhysRevLett.118.200502,HuberShadowBound,PhysRevResearch.5.033144,ning2025linear,PhysRevA.102.022413,HuberWyderkaAME}.
Explicit constructions often stem from non-generic objects such as combinatorial designs~\cite{Goyeneche_2015,PhysRevA.97.062326,zang2021quantum}, orthogonal arrays~\cite{PhysRevA.90.022316,PhysRevA.97.062326,PhysRevA.99.042332,zang2023quantum,pang2019two}, maximum distance separable codes~\cite{PhysRevA.103.022402,raissi2018optimal,huber2020quantum}, as well as related algebraic notions~\cite{facchi2010classical,di2013algebraic}. 

Given the benefits of absolutely maximal entanglement in communication, it is imperative to extend this notion to quantum states of electromagnetic waves---the primary medium we use to communicate.
This medium is governed by a pair of continuous variables (CV)---electric and magnetic fields---and such CV systems are being actively studied for both quantum communication~\cite{PirandolaTeleportation,cerf2007quantum,Usenko2025cvComms} and  computing~\cite{Bourassa2021Blueprint,Aghaee2025networkedCluster}.

In this work, we consider AME in CV systems and, furthermore, generalize to $k$-uniform CV states. In contrast to the finite-dimensional (qudit) setting, it has been shown that AME is \textit{generic} among infinitely squeezed Gaussian states ~\cite{Zhang2009MMES} (see also~\cite{Facchi2009MMES}), including CV cluster (\textit{a.k.a.} graph) states~\cite{ZhangBraunstein,Menicucci2006CVcluster,PhysRevA.79.062318,PhysRevA.68.062303}. Building on this foundational insight, our work contributes in two significant directions.

First, we derive necessary and sufficient conditions on the adjacency matrices of Gaussian CV cluster states to be AME, thereby providing \textit{explicit} families of CV AME states. These include constructions based on Cauchy, Vandermonde, totally positive, and real-block-code generator matrices---structures that have not previously appeared in the context of CV entanglement~\footnote{Though see Refs.~\cite{flammia2009optical, Chen2014cv60mode,Pfister2019cvCluster} on previous experimental proposals for generating cluster states.}. Our results highlight that CV cluster states are more powerful entanglement resources than previously appreciated, despite nearly two decades of active investigation~\cite{ZhangBraunstein,Menicucci2006CVcluster,PhysRevA.79.062318,PhysRevA.68.062303}. This new perspective complements prior work on Gaussian entanglement structure~\cite{PhysRevA.70.022318,Adesso2007thesis,Adesso_2007,eisert2008gaussian,PhysRevLett.105.030501,adesso2014continuous,armstrong2015multipartite,PhysRevA.95.010101,lami2018gaussian, Zhang2009MMES}, and connects to the well-known result that Haar-random qudit states become AME in the limit of large local dimension~\cite[Lemma~II.4]{hayden2006aspects,PhysRevLett.71.1291,Hayden2016randHolographic}. 

Second, we extend the study of CV AME and $k$-uniformity to \textit{non-Gaussian} states. We introduce so-called \textit{Zak cluster states}, constructed from discrete Zak basis states~\cite{Zak1967,Englert2006DiscrtZak, Fabre2020dblCylinder} and incorporating grid states (\textit{a.k.a.} Gottesman-Kitaev-Preskill states~\cite{Gottesman2001gkp,Brady2024gkpRvw,Pantaleoni2023ZakGKP}) as fundamental non-Gaussian resources. 

While a general characterization of AME typicality in the non-Gaussian setting remains open, our results provide, to our knowledge, the first explicit construction of non-Gaussian CV AME states. The introduction of non-Gaussianity may also offer robustness to noise in quantum information tasks that rely on AME or $k$-uniformity.  

In the rest of the manuscript, we introduce $k$-uniformity in the CV setting, define CV AME states, provide several random and explicit constructions, and outline CV versions of many aforementioned applications of AME states to quantum communication and many-body physics.

\begin{figure}[t]
\includegraphics[width=1\columnwidth]{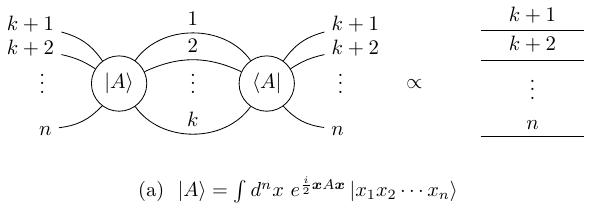}
\includegraphics[width=1\columnwidth]{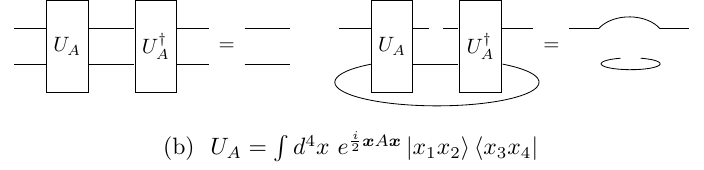}

\caption{\label{fig:tensors}
(a) Tensor-network~\cite{Pastawski_2015,Hayden2016randHolographic} depiction of \(k\)-uniformity, the condition that reduced states on any \(k\) parties are maximally mixed (i.e., proportional to the identity).
States with \(k = \lfloor n/2 \rfloor\) are absolutely maximally entangled.
(b) An \((n/2)\)-uniform state on even-\(n\) parties can be converted into a unitary by toggling half of its kets into bras. 
Uniformity translates into dual-unitarity (right diagram for \(n=4\)), i.e., contracting along ``space-like'' directions yields identity.
}
\end{figure}

\section{CV AME states}
An \(n\)-party (qubit or qudit) quantum state is \(k\)-uniform if its reduced states on any \(k\) parties are maximally mixed.
Absolutely maximally entangled states admit the highest possible uniformity at \(k=\lfloor n/2 \rfloor\). 
Two well-known (\(1\)-uniform) AME families are the \(2\)-party Bell and \(3\)-party GHZ states, existing for any qudit dimension.

CV versions of Bell and GHZ states include two-mode infinitely squeezed states (\textit{a.k.a.}\@ EPR pairs~\cite{PhysRev.47.777}) and three-mode GHZ states~\cite{PhysRevLett.84.3482,van2002quantum}, respectively.
These can be written as 
\begin{equation}\label{eq:ghz}
    |\text{GHZ}(n)\rangle = \int_{\mathbb{R}} dx~  |x\rangle^{\otimes n}    
\end{equation}
for \(n=2,3\), respectively, where \(|x\rangle\) is a single-mode position vector (i.e., an infinitely position-squeezed vacuum state displaced by \(x\)).
None of these ``states'' are normalizable, but the two- and three-mode cases can be thought of as CV analogues of qudit AME states since their single-mode reduced density matrices are all proportional to the infinite-dimensional identity matrix---the  limit of the maximally mixed state at infinite qudit dimension.

We define a CV vector to be \textit{CV \(k\)-uniform} if its reduced matrices on any \(n-k\) modes are proportional to the identity matrix, and \textit{CV AME} if \(k= \lfloor n/2 \rfloor\). 
CV \(k\)-uniform vectors are never normalizable, but there exist close approximations in the form of finitely squeezed states~\cite{PhysRevLett.112.120504,PhysRevA.83.042335,PhysRevLett.84.3482,van2002quantum,PirandolaTeleportation,PhysRevLett.95.150503}.
Such normalizable states retain the \(k\)-uniformity condition, up to small corrections in the squeezing parameter (see Supplementary Material).

\prg{CV MDS states}
It is not difficult to see that CV \(k\)-uniformity is generic.
The GHZ state~\eqref{eq:ghz} is of the form
\begin{equation}\label{eq:mds}
        \ket{G}=\int_{\mathbb{R}^{k}}d\bx\ket{\bx G}~,
\end{equation}
where the superposition is over \(n\)-mode position vectors, and where \(G\) is a real-valued \(k\times n\) matrix.
Tracing out \(k\) subsystems yields a mixed state dependent on a corresponding \(k\)-dimensional submatrix of \(G\), and this state is proportional to the identity if the submatrix has full rank.
The underlying state is CV \(k\)-uniform if and only if all such submatrices have full rank. 
Since random real matrices have full-rank submatrices \cite[pg.~321]{macwilliams1977theory}, random \(|G\rangle\) states are \textit{generically} CV AME.

The matrix \(G\) can be interpreted as a generator matrix of a classical code over the reals (\textit{a.k.a.}\@ real block code), and \(k\)-uniformity is equivalent to the underlying code being maximum distance separable (MDS) ~\cite{singleton1964maximum,roth1985generator}.
The \(1\)-uniform GHZ example~\eqref{eq:ghz} corresponds to \(G = (1,1,\cdots ,1)\)---the generator matrix of the (trivially MDS) repetition code.
Real block codes are generically MDS, and explicit constructions based on Vandermonde matrices~\cite{MarshallRealMDS} and the cosine Fourier transform~\cite{garcia2021improved} are available.
By contrast, generator matrices of binary and \(q\)-ary MDS codes, corresponding to qubit and qudit \(k\)-uniform states~\cite{PhysRevA.103.022402,raissi2018optimal,huber2020quantum}, respectively, are non-generic and subject to decades of study~\cite{macwilliams1977theory,roth2006introduction}.

\prg{CV cluster states}
Our primary examples of generic CV AME states are the CV cluster states,
\begin{equation}\label{eq:graph}
\ket{A}=\int_{\mathbb{R}^{n}}d\bx~e^{i\bx A \bx /2}\ket{\bx}~,
\end{equation}
where \(|\bx\rangle\) is an \(n\)-mode position vector, \(A\) is a real-valued weighed adjacency matrix of the state's underlying graph, and \(\bx A \bx = \sum_{j \ell} x_{j} A_{j \ell} x_{\ell}\).
These states are a CV analogue of qudit cluster states, \(\sum_{\boldsymbol{b}}\omega^{\boldsymbol{b}A'\boldsymbol{b}/2}|\boldsymbol{b}\rangle\), where \(\omega\) is a root of unity, \(\boldsymbol{b}\) is an \(n\)-dimensional vector over \(\mathbb{Z}_q\), and \(A'\) is a \(q\)-ary matrix~\cite{PhysRevLett.86.910,PhysRevA.69.062311,hein2006entanglement}.
We also define analogous cluster states on rotors~\cite{bermejo2014normalizer,Albert2017GenPhaseSpace,PhysRevA.110.022402}, whose underlying alphabet is \(\mathbb{Z}\), and demonstrate that \(k\)-uniformity exists in such states in the Supplementary Material.

Tracing out \(k \leq  n/2 \) modes yields a reduced state that depends on an ``off-diagonal'' rectangular submatrix of \(A\) whose \(k\) columns correspond to the traced-out modes and whose \(n-k\) rows correspond to the remaining modes. 
The reduced state is proportional to the identity if and only if this submatrix has full rank, and a CV cluster state is \(k\)-uniform if and only if this is true for all subsets of \(k\) modes (see Supplementary Material for proofs; cf.~\cite[Thm.~7]{HelwigGraphState} for qudit case).

If the entries of $A$ are chosen independently from a continuous distribution, any off-diagonal $k\times{}(n-k)$ submatrix of \(A\) has full rank with probability one. 
It follows that the random matrix $A$ has the property required for $k$-uniformity with probability one \textit{for any} $k\leq n/2$.
In other words, random CV cluster states are exactly AME with probability one for any fixed \(n\). 
In contrast, being exactly AME for fixed \(n\) is more difficult for qudits since \(q\)-ary matrices have only \(q\) possible entries and thus have full rank less often~\cite{HelwigGraphState,PhysRevA.106.062424,feng2017multipartite,sudevan2022n,PhysRevResearch.2.033411}.

Explicit CV AME cluster states can be constructed from many matrix families associated with complete graphs.
We present four examples of \(A\) matrices, each extendable to any number of modes:
\begin{equation}\label{eq:matlist}
\!\!\left(\begin{smallmatrix}1  &  1  &  1  &  1\\
 1  &  2  &  3  &  4\\
 1  &  3  &  6  &  10\\
 1  &  4  &  10  &  20 
\end{smallmatrix}\right),\left(\begin{smallmatrix}1  &  1/2  &  1/3\\
 1/2  &  1/3  &  1/4\\
 1/3  &  1/4  &  1/5 
\end{smallmatrix}\right),\left(\begin{smallmatrix}1  &  1  &  1\\
 1  &  1/2  &  1/4\\
 1  &  1/4  &  1/16 
\end{smallmatrix}\right),\left(\begin{smallmatrix}\sqrt{2}  &  \sqrt{3}  &  \sqrt{5}\\
 \sqrt{3}  &  \sqrt{7}  &  \sqrt{11}\\
 \sqrt{5}  &  \sqrt{11}  &  \sqrt{13} 
\end{smallmatrix}\right).
\end{equation}

The first matrix in the list~\eqref{eq:matlist} is a symmetric Pascal matrix, whose \((j,\ell)^{\text{th}}\) element is the binomial coefficient \(j+\ell \choose j\).
These elements satisfy Pascal's identity, which is used to prove that their submatrices all have positive determinant \cite[Exam.~2.2.4]{Fallat}.
This is equivalent to the desired criteria, yielding a family of CV AME states.

The second matrix in the list~\eqref{eq:matlist} is a Hilbert matrix with entries \(1/(j+\ell-1)\)---a real-valued example of a Singleton array~\cite{macwilliams1977theory,roth1985generator,maruta1990singleton,raissi2018optimal}.
It is a member of the broader class of Cauchy matrices, which have entries \(1/(v_{j}+w_\ell)\) 
for two real vectors \(\boldsymbol{v},\boldsymbol{w}\) with nonzero increasing entries.
Submatrices of Cauchy matrices are known to have positive determinant~\cite{Pinkus}.
They are further subsumed by Hankel matrices, whose \((j,\ell)^{\text{th}}\) element depends only on \(j+\ell\).
Any Hankel matrix satisfying a simple condition \cite[Thm.~4.4]{Pinkus} yields a CV AME state.
Such states should be realizable using an optical frequency comb~\cite{flammia2009optical,Chen2014cv60mode,Pfister2019cvCluster,Wang2024cvClusterOnchip}.

The third matrix in the list~\eqref{eq:matlist} is a Vandermonde matrix, generally defined by entries \(v_{j}^\ell\). 
Such a matrix yields a CV AME cluster state whenever it is symmetric and \(\boldsymbol{v}\) consists of nonzero increasing entries \cite[Sec.~4.3]{Pinkus}.

These first three matrices are totally positive, meaning any of their submatrices has positive determinant.
This is sufficient for a symmetric matrix to yield a CV AME cluster state.
Other admissible matrices include exponential kernels with elements \(\exp(v_{j} w_{\ell})\),
matrices \(u^{(j-\ell)^2}\) for some \(u\in(0,1)\) \cite[pg.~42]{Pinkus}, and any symmetric matrices for which products of certain elements increase sufficiently quickly \cite[Thm.~2.16]{Pinkus}.
Totally positive \(A\) also yield MDS generator matrices for the states in Eq.~\eqref{eq:mds} for \(n=2k\) by letting \(G = (I_k|A)\), where \(I_k\) is the \(k\)-dimensional identity.

More generally, only ``off-diagonal'' submatrices of \(A\) need to have full rank for the corresponding CV cluster states to be CV AME.
An example satisfying these less stringent criteria is the fourth matrix in the list~\eqref{eq:matlist}---a symmetric matrix whose elements are square-roots of different primes.
Non-symmetric versions of such matrices admit nonzero determinants for all of their submatrices~\cite{stack}.
The same argument applies to the ``off-diagonal'' blocks of the symmetric version, yielding another class of CV AME states.

\prg{Analog stabilizer states}
Both CV MDS and CV cluster states are special cases of infinitely squeezed Gaussian states,
which can be treated in a stabilizer formalism. 
In other words, in the limit of infinite squeezing on all modes, pure Gaussian states become analog (\textit{a.k.a.}\@ symplectic, linear, Gaussian) stabilizer states~\cite{PhysRevLett.80.4088,PhysRevLett.80.4084,barnes2004stabilizer}---\(+1\)-eigenvalue eigenstates of a group \(\mathbb{R}^n\) of commuting displacement operators (see Supplementary Material).

Analog stabilizer group elements (prepended by a phase, which we set to one for simplicity) are of the form \(\exp(i\boldsymbol{r}H\hat{\boldsymbol{\gamma}})\), where \(\br\) is any real \(n\)-dimensional row vector, \(H\) is an \((n\times 2n)\)-dimensional stabilizer generator matrix, and \(\hat{\boldsymbol{\gamma}}=(\hat{\boldsymbol{x}},\hat{\boldsymbol{p}})\) is the \(2n\)-dimensional vector of \(n\) position and \(n\) momentum quadratures.
The group's Lie (\textit{a.k.a.}\@ nullifier~\cite{PhysRevA.79.062318}) algebra is generated by rows of \(H \hat{\boldsymbol{\gamma}}\).
Commutation is ensured by the condition \(H\Omega_{2n}H^\intercal=0\), where \(\Omega_{2n}\) is the usual bosonic symplectic form, and \(\intercal\) denotes transposition.

\begin{definition}
    For a linearly independent set of operators $\{H_i\hat{\bgamma}:1\leq{}i\leq{}k\}$ satisfying $H\Omega_{2n}H^\intercal=0$, the subspace annihilated by it defines an \textnormal{analog stabilizer code}. Such a code is denoted as a $[[n,n-k]]_{\mathbb{R}}$ code. When $k=n$, the code has a single element, and is called an \textnormal{analog stabilizer state}.
\end{definition}

An analog stabilizer state can be written as
\begin{equation}\label{eq:analog}
    \ket{H}=\int_{\mathbb{R}^{n}}d\bx\,e^{i\bx AB^{\intercal}\bx/2}\ket{\bx B}~,
\end{equation}
where \(H = (A|-B)\) is split into position and momentum pieces.
CV cluster states~\eqref{eq:graph} have \(G = (A | -I_n)\).
CV MDS states~\eqref{eq:mds} have \(B = \left(\begin{smallmatrix}G\\
\boldsymbol{0}
\end{smallmatrix}\right)\) and \(A=\left(\begin{smallmatrix}\boldsymbol{0}\\
h
\end{smallmatrix}\right)\), where \(h\) is an \((n-k)\times n\) parity-check matrix of the MDS code generated by \(G\), satisfying \(h G^{\intercal} = \boldsymbol{0}\), and \(\boldsymbol{0}\) denotes the zero matrix (cf.~\cite{raissi2018optimal,raissi2020modifying} for qudit case).

Generic AME in analog stabilizer states can be derived by noting that
every stabilizer state is equivalent to a CV cluster state up to tensor-product Gaussian operations (see Supplementary Material; cf.~\cite{VdNestStabilizerGraph} for qudit case).
Since tensor-product Gaussian operations do not change a state's entanglement, and cluster states are generically AME, \(n\)-mode analog stabilizer states are as well. The fact that the infinite squeezing limit of Gaussian states generically have this maximal entanglement property has been established in \cite{Zhang2009MMES}, where such states are referred to as maximally multipartite entangled states (MMES). Our result rederives this fact, working directly with the nullifier algebra. For comparison, random \(n\)-qubit stabilizer states, while being on average close to maximally entangled over any bipartition \cite[Thm.~III.2]{PhysRevA.74.062314}\cite{PhysRevLett.125.241602}, are never AME for large $n$ \cite{scott}.

An alternative argument can be made by adapting the method to obtain \(k\)-uniformity from the stabilizer generator matrix of qubit states~\cite{majidy_scalable_nodate}.
Along similar lines, the pure (\textit{a.k.a.} diagonal) distance~\cite{PhysRevA.90.062304,PhysRevA.78.042303,wagner2022pauli}---the minimum weight of a stabilizer---is a proxy for uniformity.
Code basis states of an analog stabilizer code with pure distance \(k+1\) are all \(k\)-uniform.
For example, analog surface code~\cite{PhysRevA.78.052121} states are \(3\)-uniform, while the Braunstein five-mode code~\cite{PhysRevLett.80.4084,PhysRevA.81.062305} yields \(2\)-uniform AME states.

\prg{Zak cluster states}
While Gaussian states are easy to prepare, they cannot provide protection against Gaussian error models~\cite{Niset2008nogo}. This motivates the search of non-Gaussian AME states. We first note that in planar rotor systems, we can construct AME states using a similar stabilizer formalism. Zak's \textit{kq} representation~\cite{Zak1967} gives a basis of the bosonic Hilbert space that is analogous to rotors, known as the discrete Zak basis~\cite{Englert2006DiscrtZak} (see also Refs.~\cite{Fabre2020dblCylinder,Pantaleoni2023ZakGKP}). Employing this basis, we construct \textit{Zak cluster states},
\begin{equation}
    \ket{A,P}=\sum_{\bl,\bm\in\mathbb{Z}^n}e^{\frac{i}{2}(\bl{}A\bl+\bm{}P\bm)}\ket{\bl,\bm}
\end{equation}
We show that they are generically AME, in direct analogy with the rotor case; see Supplementary Material for technical details. To the best of our knowledge, this is the first construction of non-Gaussian AME states.  

\section{Applications}

\prg{CV teleportation}
We anticipate the primary application of CV AME states to be multi-party CV teleportation.
We review a CV analogue~\cite{Zhang2009MMES} of a qubit-based open-destination teleportation protocol~\cite{HelwigTeleportation} using CV AME states.
We expect this protocol will find use in realizing a quantum internet~\cite{Rohde2025QuInternet} via CV quantum communication links.
Related protocols include using GHZ~\eqref{eq:ghz} and related states for one-party teleportation~\cite{PhysRevLett.84.3482,van2002quantum} and telecloning~\cite{PhysRevLett.87.247901}, the analog surface code for broadcasting~\cite{PhysRevA.97.032345}, and a \(3\)-uniform \(12\)-mode state for three-party teleportation~\cite{nesterova2025tridirectional}.

An \(n\)-qubit AME state is maximally entangled across any bipartition, meaning that it can be brought into the form where there are two sets of  \(k = \lfloor n/2 \rfloor\) parties, each sharing a Bell state with its partner, and an ancillary disentangled \((n-2k)^{\text{th}}\) party (when \(n\) is odd).
The \(k\) parties then use conventional quantum teleportation  to relay qubit states.
Since this is possible for any combination, the parties can pair up after the AME state is initialized.

The CV protocol proceeds analogously using a Gaussian CV AME state, Gaussian operations, and homodyne measurements.
We outline this protocol using ideal non-normalizable states, but an actual realization will require finitely squeezed versions of the CV AME state; the imperfections of such a substitution are well documented~\cite{PhysRevLett.84.3482,van2002quantum,PirandolaTeleportation,PhysRevLett.95.150503}.

A CV AME state \(|\psi\rangle\) can be brought into a form in which \(k= \lfloor n/2 \rfloor\) pairs of modes each share an ideal two-mode squeezed state~\eqref{eq:ghz},
\begin{equation}\label{eq:activetransformation}
    U_{\sfb}|\psi\rangle\propto \int_{\mathbb{R}^{k}}d\boldsymbol{x}|\boldsymbol{x}_{\sfa},\boldsymbol{x}_{\sfb}\rangle|\phi\rangle
    =|\text{GHZ}(2)\rangle^{\otimes k}|\phi\rangle~,
\end{equation}
where \(\sfa\) consists of one member from each pair, \(\sfb\) collects the remaining members, and \(|\phi\rangle\) is an ancillary state on the remaining \((n-2k)^{\text{th}}\) mode (when \(n\) is odd).
The proportionality symbol is due to constants arising from changes of variables [e.g., \(\int_{\mathbb{R}} dx f(x/2) = 2 \int_{\mathbb{R}} dx f(x)\)], becoming an equality when finitely squeezed normalizable versions of states are used.
The required Gaussian unitary \(U_{\sfb}\) acts on \(\sfb\) and the ancillary mode only, and the \(k\) pairs then perform conventional CV teleportation from \(\sfa\) to \(\sfb\)~\cite{PhysRevLett.84.3482,van2002quantum,PirandolaTeleportation}.

\prg{Majority-agreed key distribution}
Given an $n$-mode CV AME stabilizer state, with each mode held by a separate party, majority-agreed key distribution (MAKD)~\cite{MAKD} allows for secure key distribution between any two parties given the cooperation of a majority, including themselves.

Say two parties $\sfa$ and $\sfb$ want to share a key. If $\sfs$ is a set of parties including $\sfa$ and $\sfb$ such that $|\sfs|\geq{}(n+1)/{2}$, then they can find a nullifier of the state that is independent of the quadratures of modes in possession of the complement $\sfs^\complement$ through Gaussian elimination. 
The parties in $\sfs$ measure their corresponding part of the nullifier, and everyone but $\sfa$ and $\sfb$ announces their results. 
As a result, $\sfa$ and $\sfb$ share a CV key. 
The CV AME property guarantees that the cooperation of a majority is necessary, as there cannot be a nullifier that is local to $\lfloor n/2 \rfloor$ parties or less.

\prg{Information masking \& secret sharing}
CV AME states can be used in CV analogues of quantum information masking schemes. 
An operation from a logical Hilbert space $\mathcal{H}_L$ to an $n$-body Hilbert space, $\otimes_{l=1}^n\mathcal{H}_{A_l}$, is said to $k$-uniformly mask quantum information~\cite{ModiMasking,ShiMasking} in all states if, for any reduction to $k$ parties, the reduced encoded state gives no information about which state in $\mathcal{H}_L$ was encoded. 
As proved for the qudit case~\cite{ShiMasking}, $(k+1)$-uniform states give rise to $k$-uniform masking of quantum information in all states. 
In continuous variables, all states can be $k$-uniformly masked in a $n$-mode system for $k\leq(n-1)/2$. 

CV teleportation and information masking are intimately related to constructing CV quantum secret sharing (QSS) schemes~\cite{LauCVQSS,PhysRevA.95.012315}; a $2m$-mode CV AME state gives rise to an $((m,2m-1))$ threshold CV QSS scheme.

\prg{Perfect tensors \& holography}
AME states for even \(n\) are in one-to-one correspondence with perfect tensors.
Consider an \(n\)-mode CV AME state, \(\int d\boldsymbol{x}\,T(\boldsymbol{x})|\boldsymbol{x}\rangle\), and partition its modes into two sets, \(\sfa\) and \(\sfb\), with their sizes satisfying \(|\sfa| \leq |\sfb|\).
Create an operator \(S\) that maps the smaller set into the larger one by flipping the kets pertaining to \(\sfa\) into bras. Letting \(\bx = (\boldsymbol{y},\boldsymbol{z})\),
\begin{equation}
    S_{\mathsf{A}\to\mathsf{B}}=\int_{\mathbb{R}^{|\mathsf{A}|}}d\boldsymbol{y}\int_{\mathbb{R}^{|\mathsf{B}|}}d\boldsymbol{z}|\boldsymbol{z}_{\mathsf{B}}\rangle\,T(\boldsymbol{y},\boldsymbol{z})\,\langle\boldsymbol{y}_{\mathsf{A}}|~.
\end{equation}
The CV AME condition implies that \(S\) is an isometry (up to a multiplicative constant), meaning that it preserves distances in the smaller space after embedding into the larger one.
This is true for any bipartition, and the tensor \(T(\bx)\) corresponding to the state is called \textit{CV perfect}.
CV cluster states give rise to Gaussian CV perfect tensors (see Supplementary Material).

For instance, extending Ref.~\cite{Pastawski_2015,Hayden2016randHolographic} to the CV case, we may construct a Gaussian holographic state as a tensor network by stitching together many fundamental CV perfect tensors $T$ on a hyperbolic lattice. 
Since CV AME states exist on any number of modes, both random and explicit tensor networks for \textit{arbitrary} hyperbolic lattices are possible, and averaging over random instantiations can be done without a large-\(n\) limit.

Entanglement between two boundary subsystems should obey a holographic area law that depends on the number of bulk modes crossed by a minimal geodesic passing through the bulk, multiplied by an entropy function that depends on the energy per mode~\cite{Cramer2006QHOAreaLaw,Pastawski_2015}.

\prg{Multi-unitarity}
A unitary operator $U$ is dual-unitary~\cite{BertiniProsenExact,Rather_2022} if the realigned operator $U^{R}$ defined by $\bra{j\ell}U\ket{pq}=\bra{q\ell}U^{R}\ket{pj}$ is also unitary.
More generally, tri-unitary operators~\cite{JonayTriunitary}, which are 3-to-3 isometric tensors such that rotating the tensor legs on a plane preserves unitarity, and higher-order analogues may also be defined. 
Dual- and tri- unitary operators correspond to tensors that are isometric across only certain partitions and are thus subsumed by perfect tensors.
CV AME cluster states yield a large family of Gaussian dual- and higher-unitary operators [see Fig.~\ref{fig:tensors}(b) and Supplementary Material].

\section{Conclusion}

We extend the notion of absolutely maximal entanglement (AME) from finite- to infinite-dimensional quantum systems, and find that such entanglement is generic among infinitely squeezed Gaussian states, which include continuous-variable (CV) cluster states, in line with prior work~\cite{Zhang2009MMES}. Our explicit constructions of CV AME states, based on structured families of adjacency matrices, reveal a rich landscape of multipartite entanglement and give rise to simple CV analogues of numerous applications underpinned by discrete-variable AME states.

Moving beyond the Gaussian regime, we introduce \textit{Zak cluster states}, constructed from discrete Zak basis states (related to GKP grid states) and non-Gaussian entangling gates. Our findings thus open the door to systematic study of non-Gaussian AME. Such states may also offer improved robustness to physical (e.g., Gaussian) noise---an important consideration given the no-go theorems for robust quantum information processing with only Gaussian resources~\cite{Niset2008nogo,vuillot}.

The utility of CV AME states is underscored by significant experimental progress in generating large-scale CV cluster states, including demonstrations with over a million modes in bulk optics~\cite{Yoshikawa2016million}, implementations of CV quantum secret sharing~\cite{Cai2017cvQss}, generation of two-dimensional cluster states~\cite{Asavanant2019_2Dcluster}, and more recent advances in chip-integrated~\cite{Wang2024cvClusterOnchip} and networked~\cite{Aghaee2025networkedCluster} settings.

The ubiquity of high entanglement in CV states begs the question of whether CV correlations are more powerful than those in finite dimensions. 
This has been shown for two parties~\cite{coladangelo2020inherently}, but other separations may be possible for three or more.
Related ideas include study of CV stabilizer self-testing~\cite{MAKD,makuta2021self}, Bell-type inequalities~\cite{banaszek1999nonlocality,PhysRevLett.99.210405}, and the marginal problem~\cite{eisert2008gaussian,yu2021complete}.

In future work, it would be valuable to further classify and generalize non-Gaussian CV AME states, explore alternative non-Gaussian resource states beyond GKP and Zak-type constructions, and quantitatively assess the robustness of such states under noise. Hybrid architectures blending both Gaussian and non-Gaussian features may also prove fruitful. These directions offer promising opportunities in CV quantum information processing, Gaussian and non-Gaussian resource theories, and quantum optics.

\begin{acknowledgments}
V.V.A.\@ acknowledges 
Alexander Barg,
Ben Q.\@ Baragiola,
Jonathan \@Conrad,
Steven T.\@ Flammia,
Michael J.\@ Gullans,
Daniel Gottesman,
Patrick Hayden,
and
Shayan Majidy
for helpful discussions.
A.J.B.\@ acknowledges support from the NRC Research Associateship Program at NIST.
V.V.A.~acknowledges NSF grants OMA2120757 (QLCI) and CCF2104489.
V.V.A.\@ thanks Ryhor Kandratsenia for providing daycare support throughout this work.
\end{acknowledgments}

\onecolumngrid
\appendix
\newpage
\section{Supplementary Material}

We provide proofs of our main theorems, along with discussion on infinitely squeezed states and normalizable CV cluster states. 

\subsection{I.~~~CV AME states}

\thref{uniform} is our main result, giving a necessary and sufficient condition on the matrix $A$ for the corresponding cluster state to be $k$-uniform.

\begin{theorem}\thlabel{uniform}
    For $k\leq\frac{n}{2}$, the CV cluster state corresponding to $A$ is $k$-uniform if and only if every $k\times{}(n-k)$ submatrix of $A\in{}M_{n\times{}n}(\mathbb{R})$ obtained by deleting rows in some index set $S$ and columns in $S^\complement$ has full rank.
\end{theorem}

\begin{proof}
Assume that the condition on the submatrices of $A$ hold. Fix $S\subseteq[n]$ such that $|S|=k$. We want to show that $\tr_{S^\complement}\ket{A}\bra{A}$ is proportional to the identity. Let
\begin{align}\label{eq:amatrix}
A =~ \begin{blockarray}{ccl}
  S & S^\complement & \\
\begin{block}{(cc)l}
  B & C & ~S \\
  C^\intercal & D & ~S^\complement \\
\end{block}
\end{blockarray}
\end{align}
be the block diagonal form of $A$. Then we have
\begin{subequations}
\begin{align}
\tr_{S^{\complement}}\ket{A}\bra{A}&=\tr_{S^{\complement}}\int_{\mathbb{R}^{n}}d\bx\int_{\mathbb{R}^{n}}d\bx'e^{\frac{i}{2}(\bx A\bx-\bx'A\bx')}\ket{\bx}\bra{\bx'}\\&=\int_{\mathbb{R}^{n-k}}d\by\int_{\mathbb{R}^{n}}d\bx\int_{\mathbb{R}^{n}}d\bx'e^{\frac{i}{2}(\bx A\bx-\bx'A\bx')}\braket{\by|\bx}\braket{\bx'|\by}\\&=\int_{\mathbb{R}^{n-k}}d\by\int_{\mathbb{R}^{k}}d\bx_{S}\int_{\mathbb{R}^{k}}d\bx'_{S}e^{\frac{i}{2}(2(\bx_{S}-\bx_{S}')C\by+\bx_{S}B\bx_{S}-\bx'_{S}B\bx'_{S})}\ket{\bx_{S}}\bra{\bx'_{S}}\\&=\int d\bx_{S}d\bx'_{S}e^{\frac{i}{2}(\bx_{S}B\bx_{S}-\bx'_{S}B\bx'_{S})}\int_{\mathbb{R}^{n-k}}d\by e^{i(\bx_{S}-\bx'_{S})C\by}\ket{\bx_{S}}\bra{\bx'_{S}}\ \\&=\int d\bx_{S}d\bx'_{S}e^{i(\bx_{S}B\bx_{S}-\bx'_{S}B\bx'_{S})}(2\pi)^{n-k}\delta^{(n-k)}((\bx_{S}-\bx'_{S})C)\ket{\bx_{S}}\bra{\bx'_{S}}~.
\end{align}
\end{subequations}
Since $C$, being an appropriate $k\times(n-k)$ submatrix of $A$, has full rank, $CC^\intercal$ is a positive definite $k\times{}k$ matrix. Now treat the Dirac delta as a limit of Gaussians of width \(\Delta\) and let $\bz=\sqrt{CC^\intercal}(\bx_S-\bx'_S)$. Then
\begin{subequations}
\begin{align}
\delta^{(n-k)}((\bx_S-\bx'_S)C)&=\lim_{\Delta\rightarrow{}0}\frac{1}{(\sqrt{2\pi\Delta})^{n-k}}e^{-\frac{1}{\Delta}(\bx_S-\bx'_S)CC^\intercal(\bx_S-\bx'_S)}\\
&=\lim_{\Delta\rightarrow{}0}\frac{1}{(\sqrt{2\pi\Delta})^{n-k}}e^{-\frac{1}{\Delta}\bz\cdot\bz}\\
&=\delta^{(n-k)}(\bz,\boldsymbol{0})\\
&=\frac{\delta(0)^{n-2k}}{\sqrt{\det{}CC^\intercal}}\delta^{(k)}(\bx_S-\bx'_S)~,
\end{align}
\end{subequations}
where \(\mathbf 0\) is the zero vector of \(n-2k\) dimensions such that \((\bz,\boldsymbol{0})\) is of dimension \(n-k\).
The partially traced state is 
\begin{subequations}
\begin{align}
\tr_{S^\complement}\ket{A}\bra{A}&=\frac{\delta(0)^{n-2k}(2\pi)^{n-k}}{\sqrt{\det(CC^\intercal)}}\int d\bx_Sd\bx'_Se^{i(\bx_SB\bx_S-\bx'_SB\bx'_S)}\delta^{(k)}(\bx_S-\bx'_S)\ket{\bx_S}\bra{\bx'_S}\\
&=\frac{\delta(0)^{n-2k}(2\pi)^{n-k}}{\sqrt{\det(CC^\intercal)}}\int{}d\bx_S\ket{\bx_S}\bra{\bx_S}~.
\end{align}
\end{subequations}
If there is a $k\times(n-k)$ submatrix of $A$ with rows in $S^\complement$ and columns in $S$ that does not have full rank, there will be nonzero cross terms in the partially traced state due to $\delta^{(n-k)}((\bx_S-\bx'_S)C)$ not being proportional to $\delta^{(k)}(\bx_S-\bx'_S)$. So the inverse holds.

\end{proof}

\subsection{II.~~~Analog stabilizer states $=$ infinitely squeezed Gaussian pure states}

As we will later prove in \thref{Stabgraph}, uniformity extends to general analog stabilizer states, whose form we determine below.
See \cite[Appx. E]{vuillot} for the more general case of the encoding matrix for analog stabilizer codes.

The analog stabilizer states we study can be obtained from the zero position state, \(|\bx=\boldsymbol{0}\rangle\), via a displacement-free Gaussian circuit.
They can be thought of as infinitely squeezed versions of pure non-displaced Gaussian states, which are obtained via the same circuit from the vacuum.

Any pure non-displaced Gaussian state can be characterized by the following form of its \(2n\)-dimensional covariance matrix \cite[Prob.~5.6]{serafini2023quantum},
\begin{equation}
    \boldsymbol{\sigma} = O^\intercal D O \quad\quad\text{with}\quad\quad D=\begin{pmatrix}\Lambda & \boldsymbol{0}\\
\boldsymbol{0} & \Lambda^{-1}
\end{pmatrix}~,    
\end{equation}
where \(O\) is an orthogonal symplectic matrix that performs a rotation on the \(2n\) quadratures, and where \(\Lambda \geq \boldsymbol{0}\) is a diagonal matrix whose eigenvalues characterize the degree of squeezing of the quadratures.

In the limit of infinite squeezing, the variance in half of all quadratures vanishes, and the variance in the other half goes to infinity.
As such, the top left block of \(D\) represents the infinite-variance quadratures, and the bottom right block represents the zero-variance ones.
The latter nullifying quadratures define the state and are obtained by acting by the \((n\times 2n)\)-dimensional bottom part of \(O\) on the original quadrature vector \(\hat{\bgamma}=\begin{pmatrix}
    \hat{\bx}\\
    \hat{\bp}
\end{pmatrix}\), yielding
\begin{equation}
    H = O_2\quad\quad\text{where}\quad\quad O=\begin{pmatrix}O_{1}\\
O_{2}
\end{pmatrix}~.
\end{equation}

Rows of the stabilizer generator matrix \(H\) form the basis vectors of a ``symplectic'', or ``compatible'', hyperplane whose generator matrix satisfies
\begin{equation}
    H\Omega_{2n}H^\intercal=0~,\quad\quad\text{where}\quad\quad \Omega_{2n}=\left(\begin{smallmatrix}0 & I_{n}\\
-I_{n} & 0
\end{smallmatrix}\right)
\end{equation}
is the bosonic symplectic form.
The Lie algebra of the stabilizer group, which consists of the nullifiers of the state, is generated by the entries of \(H \boldsymbol{\hat{\gamma}}\).

For each state, orthogonal Gaussian transformations rotate the \(n\) real quadratures into each other and leave the hyperplane state invariant.
Such transformations form an orthogonal subgroup \(O(n)\) of symplectic orthogonal \(2n\)-dimensional transformations, which in turn form the group \(\text{Sp}(2n,\mathbb{R}) \cap O(2n) \cong U(n)\)~\cite[Prob.~5.6]{serafini2023quantum}.
As such, the manifold of compatible hyperplanes, also known as the Lagrangian Grassmannian~\cite{Arnold1967}, is the homogeneous space $U(n)/O(n)$.

In general, any group element can be multiplied by an arbitrary phase.
Extending to non-unity phases can be done by applying displacements, and such states can be created by applying a displacement-free Gaussian circuit to displaced position states.
The resulting set of analog states encompasses all infinitely squeezed pure Gaussian states. In this case, states are eigenstates of the group's Lie algebra with a possibly non-zero eigenvalue which characterizes the amount of displacement.
Our results about the ubiquity of AME hold for the more general displaced case since displacements are tensor-product operations that do not change entanglement.
For simplicity, we assume all stabilizer-group phases are unity from now on, with the more general case yielding displaced versions of the states we study below.

With all phases of stabilizer generators set to unity, analog stabilizer states can be written in the following form.

\begin{theorem}
    For a matrix $H=\begin{pmatrix}
        A & -B
    \end{pmatrix}$ with $A,B\in{}M_{n\times{}n}(\mathbb{R})$ such that $H$ has full rank and $H\Omega_{2n}H^\intercal=0$, the analog stabilizer state nullified by the Lie algebra generated by $\{H_i\hat{\bgamma}:1\leq{}i\leq{}n\}$ has the form 
    \begin{align}
        \ket{H}=\int_{\mathbb{R}^n} d\br\ket{\br{}B}e^{\frac{i}{2}\br{}AB^{\intercal}\br}~.
    \end{align}
\end{theorem}

\begin{proof}
The stabilizer group is given as $
    \{e^{i\br{}H\hat{\bgamma}}:\br\in\mathbb{R}^n\}
$ for $\hat{\bgamma}$.
To determine the Lie algebra structure of the stabilizer group, we compute the commutators 
\begin{subequations}
\begin{align}
    [\br{}B\hat{\bp},\bs A\hat{\bx}]&=\sum_{l,m}[(B^{\intercal}\br)_{l}\hat{p}_{l},(A^{\intercal}\bs )_{m}\hat{x}_{m}]\\&=\sum_{l,m}(B^{\intercal}\br)_{l}(A^{\intercal}\bs )_{m}(-i\delta_{l,m})\\&=-i\bs  AB^{\intercal}\br~.
\end{align}
\end{subequations}
Note that the condition $H\Omega_{2n}H^\intercal=0$ ensures that the stabilizer group is abelian. 
Continuous linear combinations of all stabilizer group elements form what is akin to a codespace projector, which squares to itself up to infinities of the form \(\delta(0)\):
\begin{subequations}
\begin{align}
    \int d\br{}e^{i\br{}H\hat{\bgamma}}&=\int d\br{}e^{-\frac{1}{2}[\br{}B\hat{\bp},\br{}A\hat{\bx}]}e^{-i\br{}B\hat{\bp}}e^{i\br{}A\hat{\bx}}\\&=\int d\br{}e^{\frac{i}{2}\br{}AB^{\intercal}\br}e^{-i\br{}B\hat{\bp}}e^{i\br{}A\hat{\bx}}~.
\end{align}
\end{subequations}
Applying this operator to the zero position state and simplifying yields the desired state.
\end{proof}

\thref{Stabgraph} establishes that all analog stabilizer states are equivalent to some cluster state up to local Gaussian operations, and thus the uniformity of their entanglement can be characterized via \thref{uniform}.

\begin{theorem}\thlabel{Stabgraph}
    Every analog stabilizer state can be transformed into a cluster state via tensor-product Gaussian operations.
\end{theorem}
\begin{proof}
    Cluster states are simply stabilizer states for $H=\begin{pmatrix}
        A & -B
    \end{pmatrix}$ with invertible $B$. 
    If $B$ is invertible, 
    through the left action of $B^{-1}$ on $H$, we may rewrite $H=\begin{pmatrix}
        B^{-1}A & -I_n
    \end{pmatrix}$. 
    This corresponds to a linear change of variables in the integral superposition of the state.
    Also, $AB^\intercal$ must be symmetric since $H\Omega_{2n}H^\intercal=0$, and thus so is $B^{-1}(AB^\intercal)(B^{-1})^\intercal=B^{-1}A$. 

    Next, we consider stabilizer states such that the rank of $B$ is $k<n$. By left actions by matrices in $GL(n,\mathbb{R})$, which preserve the stabilizer state, we can transform $H$ into
    \begin{equation}
        H'=\left(\begin{array}{c|c}
           A_1 & B_1 \\
           A_2 & 0
        \end{array}\right)
    \end{equation}
    where $A_1$, $B_1$ are $k\times{}n$ matrices, $A_2$ is a $(n-k)\times{}n$ matrix, and $B_1$ has rank $k$. Since $H'$ has full rank, so does $A_2$. Choose a set $S\subseteq[n]$ with $|S|=k$ such that the columns of $B_1$ indexed by $S$ are linearly independent. Let $A_S$,$A_{S^\complement}$,$B_S$,$B_{S^\complement}$
    be submatrices of $A_2$ and $B_1$ obtained from the columns in the corresponding index sets. Since $H'\Omega_{2n}H'^\intercal=0$, we must have $A_2B_1^\intercal=A_SB_S^\intercal+A_{S^\complement}B_{S^\complement}^\intercal=0$. And since $S$ was chosen so that $B_S$ is invertible, we have
    \begin{equation}
        A_S=-A_{S^\complement}B_{S^\complement}^\intercal(B_S^{\intercal})^{-1}
    \end{equation}
    So the column space of $A_S$ is a subspace of the column space of $A_{S^\complement}$. Since $A_2$ has full rank, $A_{S^\complement}$ must be invertible.
    
    After the Fourier operation on the modes corresponding to $S^\complement$, the matrix $H'$ is transformed so that the right half is
    \begin{align}
    \begin{blockarray}{cc}
    ~S & S^\complement & \\
    \begin{block}{(cc)}
    ~B_S & C\\
    ~0 & A_{S^\complement}\\
    \end{block}
    \end{blockarray}
    \end{align}
    for some $C$. This matrix is invertible since $B_S$ and $A_{S^\complement}$ are.
\end{proof}

\subsection{III.~~~Finite squeezing}

The analog stabilizer states we work with are not physical, as they cannot be normalized. 
Since we have established that any stabilizer state is equivalent to some cluster state up to local Gaussian operations, we restrict attention to cluster states and compute their Gaussian regularized form. 

Cluster states are obtained by applying the Gaussian operation $e^{i\hat{\bx}{}A\hat{\bx}/2}$ to the product of infinitely-squeezed zero-momentum eigenstates,
\begin{equation}
    |A\rangle=e^{i\hat{\bx}A\hat{\bx}/2}\int_{\mathbb{R}^{n}}d\boldsymbol{x}|\boldsymbol{x}\rangle=e^{i\hat{\bx}A\hat{\bx}/2}|\boldsymbol{p}=\boldsymbol{0}\rangle~,
\end{equation}
where each mode's momentum state is defined by \(|p\rangle=\int_{\mathbb{R}} dxe^{ipx}|x\rangle\).
Their approximate versions substitute finitely squeezed versions of the ideal momentum eigenstates.
Our squeezing parameter is \(\Delta \geq 0\), with \(\Delta = 0\) denoting the ideal non-normalizable limit.
In the limit of small \(\Delta\), such squeezing can be interpreted as the ideal states being damped by an exponential of the \(n\)-mode occupation number operator \(\hat N\). 
Using \cite[Eq. (B3)]{PhysRevLett.112.120504} yields
\begin{subequations}
\begin{align}
    e^{-\Delta(\hat{N}+n/2)}|\boldsymbol{p}=\boldsymbol{0}\rangle&\sim e^{-\Delta\|\hat{\boldsymbol{x}}\|^{2}/2}e^{-\Delta\|\hat{\boldsymbol{p}}\|^{2}/2}|\boldsymbol{p}=\boldsymbol{0}\rangle\\&=e^{-\Delta\|\hat{\boldsymbol{x}}\|^{2}/2}|\boldsymbol{p}=\boldsymbol{0}\rangle\\&=\int_{\mathbb{R}^{n}}d\boldsymbol{x}~e^{-\Delta\|\boldsymbol{x}\|^{2}/2}|\boldsymbol{x}\rangle
\end{align}
\end{subequations}
In this limit, the resulting normalized version of the CV cluster state is then
\begin{align}
\ket{\tilde{A}}\sim\frac{1}{(\Delta\sqrt{\pi})^{n/2}}\int d\bx\,e^{i\bx A\bx/2}e^{-\Delta\|\bx\|^{2}/2}{}\ket{\bx}~.
\end{align}
There exist standard techniques for working with such states \cite{PhysRevLett.112.120504,PhysRevA.83.042335,PhysRevLett.84.3482,van2002quantum,PirandolaTeleportation,PhysRevLett.95.150503}.

We verify that the state is $k$-uniform in the limit $\Delta\rightarrow{}0$.

Similar to the exact case, fix $S\subseteq[n]$ with $|S|=k$ and let $A$ be the block matrix from Eq.~\eqref{eq:amatrix}. 
Ignoring normalization,
\begin{subequations}
\begin{align}
\tr_{S^\complement}\ket{\tilde{A}}\bra{\tilde{A}}=&\int\int{}d\bx_Sd\bx_S'\ket{\bx_S}\bra{\bx_S'}e^{\frac{i}{2}(\bx_S{}B\bx_S-{\bx'_S}B\bx'_S)}e^{-\frac{\Delta}{2}({\bx_S\cdot\bx_S+\bx'_S\cdot\bx'_S)}}\\
&\times\int{}d\by{}e^{i(\bx_S-\bx'_S)C\by}e^{-\Delta\by\cdot\by}~.
\end{align}
\end{subequations}
The innermost Gaussian integral can be evaluated exactly using contour integration,
\begin{align}
    \int{}d\by{}e^{i(\bx_S-\bx'_S)C\by}e^{-\Delta{}\by\cdot\by}
    =\left(\frac{\pi}{\Delta}\right)^{\frac{n-k}{2}}\exp\left[ -\frac{1}{4\Delta}(\bx_S-\bx'_S)CC^\intercal(\bx_S-\bx'_S) \right]~.
\end{align}
So in the limit $\Delta\rightarrow{}0$, the off-diagonal elements in the reduced density operator
\begin{subequations}
\begin{align}
    \tr_{S^\complement}\ket{\tilde{A}}\bra{\tilde{A}}=&\int\int{}d\bx_Sd\bx_S'\ket{\bx_S}\bra{\bx_S'}e^{\frac{i}{2}(\bx_S{}B\bx_S-{\bx'_S}B\bx'_S)}e^{-\frac{\Delta}{2}({\bx_S\cdot\bx_S+\bx'_S\cdot\bx'_S)}}\\
    &\times(\frac{\pi}{\Delta})^{\frac{n-k}{2}}e^{-\frac{1}{4\Delta}(\bx_S-\bx'_S)CC^\intercal(\bx_S-\bx'_S)}~,
\end{align}
\end{subequations}
will be exponentially suppressed with respect to \(1/\Delta\).

\begin{figure}[t]
\centering
  \begin{subfigure}[h]{0.259\textwidth}
      \centering
      \vspace{4mm}
      \caption{}
      \includegraphics[width=\textwidth]{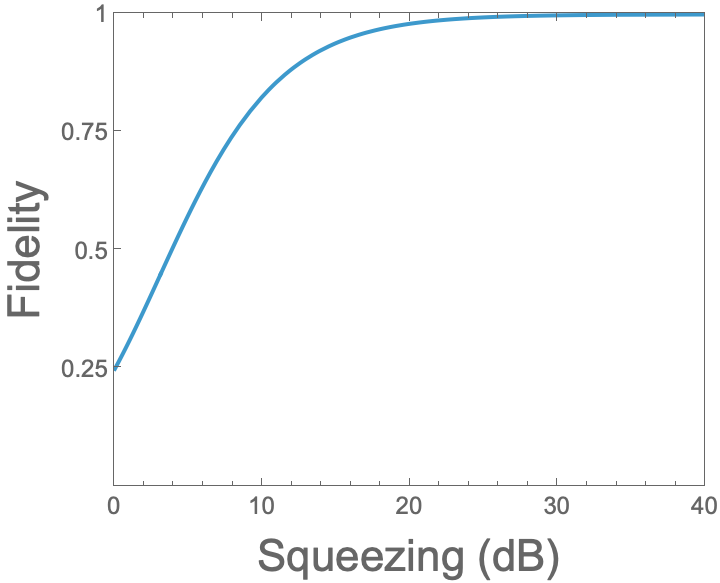}
  \end{subfigure}
  \begin{subfigure}[h]{0.22\textwidth}
    \centering
    \caption{\vspace{0.7mm}}
    \includegraphics[width=\textwidth]{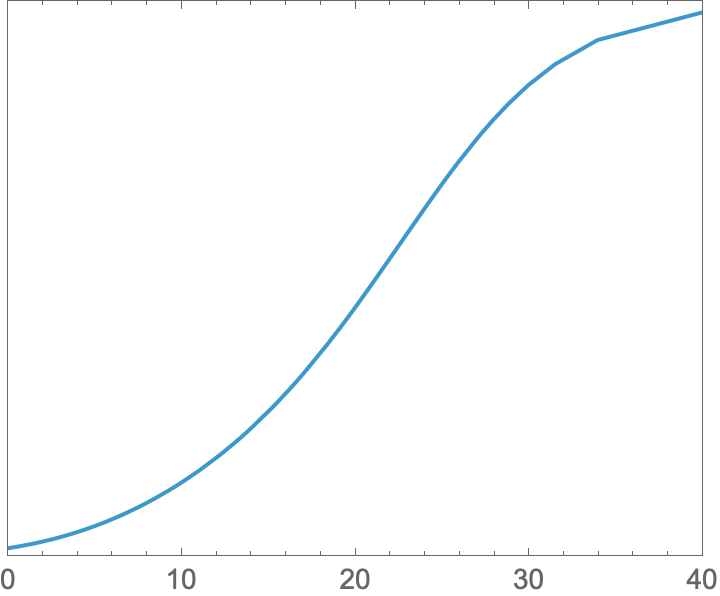}
  \end{subfigure}
  \begin{subfigure}[h]{0.22\textwidth}
    \centering
    \caption{\vspace{0.7mm}}
    \includegraphics[width=\textwidth]{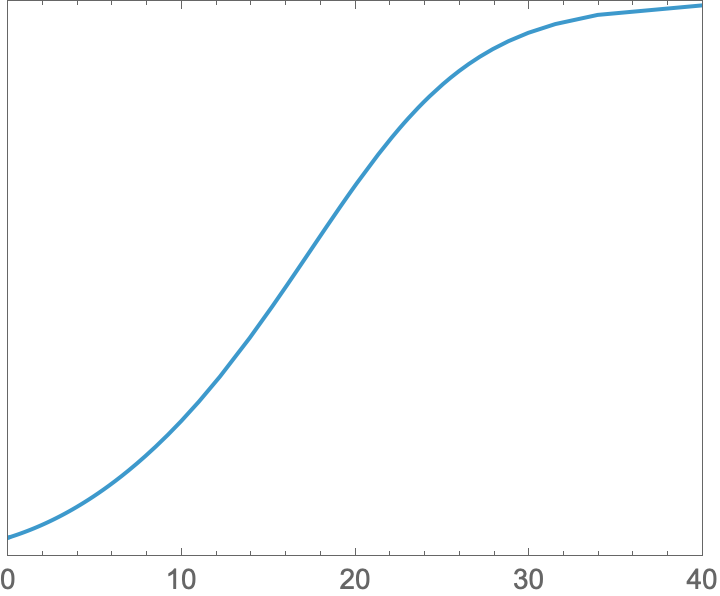}
  \end{subfigure}
  \begin{subfigure}[h]{0.22\textwidth}
    \centering
    \caption{\vspace{0.7mm}}
    \includegraphics[width=\textwidth]{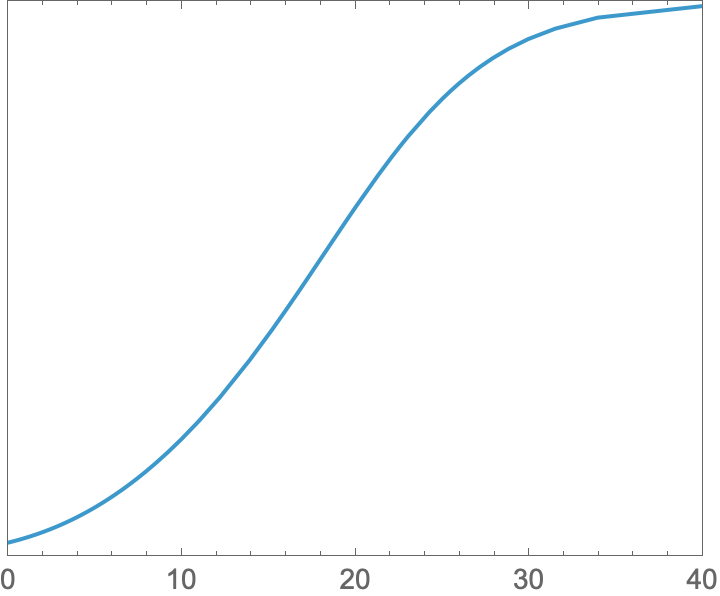}
  \end{subfigure}

\caption{\label{fig:fidelities}
Teleportation fidelity of coherent states in the two-mode protocol for cluster state matrices (a) $\left(\begin{smallmatrix}0 & 0 & 1 & 0\\
0 & 0 & 0 & 1\\
1 & 0 & 0 & 0\\
0 & 1 & 0 & 0
\end{smallmatrix}\right)$, (b) $\left(\begin{smallmatrix}1 & 1 & 1 & 1\\
1 & 2 & 3 & 4\\
1 & 3 & 6 & 10\\
1 & 4 & 10 & 20
\end{smallmatrix}\right)$, (c) $\left(\begin{smallmatrix}1 & 1 & 1 & 1\\
1 & 1/2 & 1/4 & 1/8\\
1 & 1/4 & 1/16 & 1/64\\
1 & 1/8 & 1/64 & 1/512
\end{smallmatrix}\right)$, (d) $\left(\begin{smallmatrix}0 & 1 & 1 & 2\\
1 & 0 & 1 & 3\\
1 & 1 & 0 & 3\\
2 & 3 & 3 & 0
\end{smallmatrix}\right)$ plotted against the uniform squeezing in decibels. Note that teleportation fidelity is dependent on the bipartition, as shown in~\cite{Lupo2011Frustration}. In our examples, the first two modes are chosen to be on one end, and the last two on the other. To achieve $0.999$ fidelity we require squeezing of (a) $33.007$ dB, (b) $52.712$ dB, (c) $47.710$ dB, (d) $48.582$ dB.}
\end{figure}

\subsection{IV.~~~Teleportation Fidelity}

An alternative proof of the approximate AME property can be given by showing that the fidelity of the CV teleportation protocol approaches unity as squeezing increases. As in previous works~\cite{van2002quantum,serafini2023quantum} we consider the teleportation of coherent states as a benchmark. 

For ease of calculation, we approximate cluster states via Gaussian states with uniform squeezing parameters, in contrast to the approximation via a Gaussian envelope considered in the previous section. Figure~\ref{fig:fidelities} plots the teleportation fidelity of the 2-mode teleportation protocol based on some example stabilizer hyperplanes against the amount of squeezing. For each cluster state matrix $A$ and squeezing parameter $e^r$, the approximate cluster state is a zero mean Gaussian state with covariance matrix
\begin{equation}
    \begin{pmatrix}
        (I+A^2)^{-\frac{1}{2}}&-(I+A^2)^{-\frac{1}{2}}A\\
        (I+A^2)^{-\frac{1}{2}}A&(I+A^2)^{-\frac{1}{2}}
    \end{pmatrix}
    \begin{pmatrix}
        e^r{}I&0\\0&e^{-r}I
    \end{pmatrix}
    \begin{pmatrix}
        (I+A^2)^{-\frac{1}{2}}&(I+A^2)^{-\frac{1}{2}}A\\
        -(I+A^2)^{-\frac{1}{2}}A&(I+A^2)^{-\frac{1}{2}}
    \end{pmatrix}~,
\end{equation}
where we use Serafini's convention for the covariance matrix, i.e., $e^r$ instead of $\frac{e^r}{2}$~\cite{serafini2023quantum}.
The Gaussian transformation described in~\eqref{eq:activetransformation} act as a congruence transformation on the covariance matrix~\cite[Prob.~5.6]{serafini2023quantum}. 
The resulting covariance matrix has the variance of the EPR stabilizer quadratures $\hat{x}_1-\hat{x}_3,\hat{x}_2-\hat{x}_4,\hat{p}_1+\hat{p}_3,\hat{p}_2+\hat{p}_4$ suppressed as squeezing increases. Since CV teleportation is a Gaussian protocol, the fidelity for coherent state input can be computed analytically~\cite{van2002quantum}.

The fact that teleportation fidelity approaches unity with increasing squeezing is not in contradiction with Ref.~\cite{Lupo2011Frustration}, where it is determined that finitely squeezed CV AME states cannot be locally thermal. 
We also see that the purity varies between different bipartitions of modes for fixed squeezing, but our results are only concerned with the purity eventually approaching zero, corresponding to a teleportation fidelity of unity. 

Figure~\ref{fig:fidelities}(a) shows the case where we obtain EPR pairs via passive transformations, preserving uniformity of squeezing. We see that in other examples the fidelities are lower for fixed squeezing, but approach unity nonetheless.

\subsection{V.~~~Rotor AME states}
A similar analysis can be done with planar rotor systems~\cite{bermejo2014normalizer,Albert2017GenPhaseSpace,PhysRevA.110.022402}, where the Hilbert space has a canonical basis indexed by $\mathbb{Z}$. 
We define a \textit{rotor cluster state} for a symmetric matrix $C\in{}M_{n\times{}n}(\mathbb{R})$ to be
\begin{equation}
    \ket{C}_{\mathbb{Z}}=\sum_{\bn\in{}\mathbb{Z}^n}e^{i \bn{}C\bn / 2}\ket{\bn}~.
\end{equation}
Note that the entries of $C$ are effectively elements of $\mathbb{R}/2\pi\mathbb{Z}$, and $C$ thus describes a group homomorphism $\psi_{C}:\mathbb{Z}^n\rightarrow\hat{\mathbb{Z}}^n\cong{}U(1)^n$ where $\hat{\mathbb{Z}}$ is the character group of $\mathbb{Z}$. The proof of \thref{uniform} works out to give $k$-uniformity if and only if for any set $S\in{}[n]$ with $|S|=k$, $\pi_{S^\complement}\circ\psi_C\restriction_S$ is injective. Here $\pi_{S^\complement}$ denotes the projection onto the subgroup indexed by $S^\complement$, and $\restriction_S$ the restriction of the map to the subgroup indexed by $S$. We remark that rotor AME states exist for any number of rotors, as a symmetric Vandermonde matrix with integer entries describes an appropriate homomorphism. These rotor cluster states can also be viewed as stabilizer states, with the nullifier algebra being modules over $\mathbb{Z}$, and the matrix $H=\begin{pmatrix}A & -B\end{pmatrix}$ being defined up to left action by $GL(n,\mathbb{Z})$.

Another example of a rotor AME state is any codeword of the five-rotor code \cite{faist2020continuous}---an analogue of the five-mode and five-qubit codes.
Its approximate form is treated in the Supplementary Material of Ref.~\cite{faist2020continuous}.

\subsection{VI.~~~Zak cluster states}

Inspired by the examples of rotor AME states, we present non-Gaussian AME states derived from Zak states~\cite{Zak1967,Englert2006DiscrtZak}. The Zak transform gives a basis of the bosonic Hilbert space indexed by points on a torus, which is called the Zak patch. The continuous Zak basis, denoted as $\ket{\alpha,\beta}$, consists of displaced GKP states~\cite{Gottesman2001gkp,Pantaleoni2023ZakGKP}, while its dual basis, under the double Fourier transform, is the discrete Zak basis indexed by pairs of integers $\ket{\ell,m}$~\cite{Englert2006DiscrtZak,Fabre2020dblCylinder}.

Imposing $2\pi$ periodicity, we write the continuous Zak basis states in the position basis as,
\begin{equation}\label{eq:zakbasis}
    \ket{\alpha,\beta}=\sum_{n\in\mathbb{Z}}e^{i\alpha n}\ket{x=\sqrt{2\pi}\big(n+\beta/2\pi\big)},
\end{equation}
where $\alpha,\beta\in[0,2\pi]$. The continuous Zak basis states are related to the discrete Zak basis states through the double discrete Fourier transform~\cite{Englert2006DiscrtZak},
\begin{equation}\label{eq:discrt-zakbasis}
    \ket{\alpha, \beta}=\sum_{\ell,m=-\infty}^\infty e^{-i(\ell\alpha-m\beta)}\ket{\ell,m}.
\end{equation}
Notably, the Zak basis may be effectively interpreted as two copies of the basis of the rotor Hilbert space.

We may apply a similar construction used for the rotor case to obtain non-Gaussian $k$-uniform Zak cluster states;
\begin{equation}
    \ket{A,P}=\sum_{\bl,\bm\in\mathbb{Z}^n}e^{\frac{i}{2}(\bl{}A\bl+\bm{}P\bm)}\ket{\bl,\bm}
\end{equation}
is $k$-uniform if and only if the appropriately sized off-diagonal blocks of both $A$ and $P$ have full rank. The stabilizer formalism also extends to such ``Zak cluster states"; $\ket{A,P}$ is nullified by the $\mathbb{Z}$-module generated by $\{\sqrt{2\pi}\hat{x}_j-(A\bl)_j:j\in[n]\}\cup\{\sqrt{2\pi}\hat{p}_k-(P\bm)_k:k\in[n]\}$. We remark that these Zak cluster states may be prepared by applying (non-Gaussian) entangling gates of the form $e^{\frac{i}{2}(\bl{}A\bl+\bm{}P\bm)}$ to the initial GKP resource state $\sum_{\bl,\bm}\ket{\bl,\bm}$, which is a uniform superposition over the discrete Zak basis; see Eqs.~\eqref{eq:zakbasis} and~\eqref{eq:discrt-zakbasis}. Note that this entangling gate is quadratic in the integer operators labeled by $\ell$ and $m$~\cite{Englert2006DiscrtZak,Ketterer2016ModVars,Fabre2020dblCylinder,Pantaleoni2023ZakGKP,Kutyniok1998Zakzeros,Enstad2019BalianLow}.
Similarly to the rotor case \cite{faist2020continuous}, normalizable approximations of the Zak cluster states can be shown to be approximately AME.
\subsection{VII.~~~Gaussian multi-unitarity}

We show that a CV AME cluster state on an even number of modes yields a Gaussian multi-unitary operation.
For simplicity, we consider the partition into two subsets \(\sfa\) and \(\sfb\) of \(n/2\) modes.
This yields the following partition on the state's adjacency matrix,
\begin{align}
A =~ \begin{blockarray}{ccl}
  \sfb & \sfa & \\
\begin{block}{(cc)l}
  B & C & ~\sfb \\
  C^\intercal & D & ~\sfa \\
\end{block}
\end{blockarray}~,
\end{align}
for square blocks \(B,C,D\). 
The explicit form of the unitary mapping for \((n/2)\)-dimensional vectors \(\by,\bz\) is then
\begin{subequations}
\begin{align}
T&=\int{}d\bz{}d\by{}e^{i(\bz,\by)A(\bz,\by)/2}\ket{\bz}\bra{\by}\\&=\int{}d\bz{}d\by{}e^{i(\bz{}B\bz+2\bz{}C\by+\by{}D\by)/2}\ket{\bz}\bra{\by}\\&=e^{i\hat{\bx}B\hat{\bx}/2}\cdot\int{}d\bz{}d\by{}e^{i\bz{}C\by}\ket{\bz}\bra{\by}\cdot{}e^{i\hat{\bx}D\hat{\bx}/2}\\&=e^{i\hat{\bx}B\hat{\bx}/2}\cdot\int{}d\by\ket{\bp=C\by}\bra{\by}\cdot{}e^{i\hat{\bx}D\hat{\bx}/2}\\&=e^{i\hat{\bx}B\hat{\bx}/2}\cdot\hat{F}\cdot\int{}d\by\ket{C\by}\bra{\by}\cdot{}e^{i\hat{\bx}D\hat{\bx}/2}
~.
\end{align}
\end{subequations}
In the third equality, we peel off quadratic position dependencies into Gaussian operations acting on each side.
In the fourth equality, we convert the resulting position states into momentum states.
In the last equality, we convert momentum states into position states by a Fourier operator \(\hat F\).
All operations are Gaussian, completing the proof.

\bibliography{ref}

\end{document}